\documentclass[a4paper]{amsart}
\usepackage{amsmath,amsfonts,amssymb,amsthm,graphics}
\usepackage{vaucanson-g}

\DeclareMathOperator{\rep}{rep}
\def\eint#1#2{[\![#1,#2]\!]}
\theoremstyle{plain}

\newtheorem{theorem}{Theorem}
\newtheorem{lemma}[theorem]{Lemma}
\newtheorem{corollary}[theorem]{Corollary}
\newtheorem{proposition}[theorem]{Proposition}
\theoremstyle{definition}
\newtheorem{definition}[theorem]{Definition}
\newtheorem{remark}[theorem]{Remark}

\newtheorem{example}[theorem]{Example}

\title{Deciding game invariance}

\author{Eric Duch\^ene}
\address{Universit\'{e} Lyon 1, LIRIS, UMR5205, F-69622, France,
{\tt eric.duchene@univ-lyon1.fr}}
\author{Aline Parreau}\thanks{The second author is supported by a FNRS post-doctoral grant at the University of Li{\`e}ge.}
\author{Michel Rigo}
\address{Department of Mathematics, University of Li{\`e}ge, Grande traverse 12 (B37), B-4000 Li{\`ege}, Belgium. {\tt M.Rigo@ulg.ac.be}}

\begin{document}

\begin{abstract}
In \cite{DucRig}, Duch\^ene and Rigo introduced the notion of invariance for take-away games on heaps. Roughly speaking, these are games whose rulesets do not depend on the position. Given a sequence $S$ of positive tuples of integers, the question of whether there exists an invariant game having $S$ as set of $\mathcal{P}$-positions is relevant. In particular, it was recently proved by Larsson et al. \cite{Lar} that if $S$ is a pair of complementary Beatty sequences, then the answer to this question is always positive. In this paper, we show that for a fairly large set of sequences (expressed by infinite words), the answer to this question is decidable.
\end{abstract}

\maketitle

\section{Introduction}

Let $n\ge 1$ be an integer. In this paper, we consider take-away impartial games played over $n$ piles of tokens. Two players alternatively remove a positive number of tokens from one or several piles following a prescribed ruleset. The rules are the same for both players. We assume normal convention, i.e., the player making the last move wins. Since we always remove a positive number of tokens, the game is acyclic and there is always a winner.

A {\em position} of such a game is an $n$-tuple of non-negative integers which corresponds to the number of tokens available in each pile. A {\em move} is also an $n$-tuple of non-negative integers corresponding to the number of tokens that are removed from each pile. Let $\mathbf{p}=(p_1,\ldots,p_n)$ be a position and $\mathbf{m}=(m_1,\ldots,m_n)$ be a non-zero move. The move $\mathbf{m}$ can be applied to the position $\mathbf{p}$ provided that $\mathbf{m}\le \mathbf{p}$, i.e., for all $i$, $m_i\le p_i$. The position resulting of the application of $\mathbf{m}$ is the $n$-tuple $\mathbf{p}-\mathbf{m}$.

\begin{definition}
A {\em game}, played over $n$ piles, is given by a function $G:\mathbb{N}^n\to 2^{\mathbb{N}^n}$ that maps every position $\mathbf{p}$ to a set of moves that can be chosen from $\mathbf{p}$ by the player. Otherwise stated, the ruleset is provided by the map $G$. For a position $\mathbf{p}$, the set of {\em options} of $\mathbf{p}$ is the set $\{\mathbf{p}-\mathbf{m} \mid \mathbf{m}\in G(\mathbf{p})\}$ of positions where the player can move directly.
A {\em strategy} consists in choosing a particular option for every position.
\end{definition}
An interval of integers is denoted by $\eint{k}{\ell}$.
For an example of take-away game, the game of Nim over $2$ piles is described by the map $$G_{\mathrm{NIM}}:\mathbb{N}^2\to 2^{\mathbb{N}^2}, (x,y)\mapsto \{(i,0)\mid i\in \eint{1}{x}\}\cup \{(0,j)\mid j\in\eint{1}{y}\}.$$
For Wythoff's game, the description is given by $$G_{\mathrm{WYTHOFF}}:\mathbb{N}^2\to 2^{\mathbb{N}^2},
(x,y)\mapsto G_{\mathrm{NIM}}(x,y)\cup \{(k,k)\mid k\in\eint{1}{\min \{x,y\}}\}.$$
With such a formal presentation, we recall the notion of invariant game introduced in \cite{DucRig}. Note that we shall later on distinguish two notions of invariance: invariant games and invariant subsets.
\begin{definition}
A game $G:\mathbb{N}^n\to 2^{\mathbb{N}^n}$ is {\em invariant} if there exists a set $I\subseteq\mathbb{N}^n$ such that, for all positions $\mathbf{p}$, we have
$$G(\mathbf{p})=I\cap \{\mathbf{m}\in\mathbb{N}^n\mid \mathbf{m}\le \mathbf{p}\}.$$
Otherwise stated, we may apply exactly the same moves to every position, with the only restriction that there are enough tokens left. Since a game is defined by its moves, formally by the map $G$, one also speaks of {\em invariant moves}.
\end{definition}

A motivation to introduce the notion of invariance is the relative simplicity of the corresponding rulesets. Roughly speaking, one has ``just'' to remember the set $I$.

The game of Nim defined above is invariant. Simply consider the set
$$I_{\mathrm{NIM}}=\{(i,0)\mid i\ge 1\}\cup \{(0,j)\mid j\ge 1\}.$$
Similarly, Wythoff's game is invariant with the set
$$I_{\mathrm{WYTHOFF}}=I_{\mathrm{NIM}}\cup\{(k,k)\mid k\ge 1\}.$$

For an example of non-invariant game, consider the following map,
$$G_{\mathrm{EVEN}}:\mathbb{N}^2\to 2^{\mathbb{N}^2}, (x,y)\mapsto
\begin{cases}
    \{(i,0)\mid i\in\eint{1}{x}\}, & \text{ if }x+y\text{ is even};\\
    \{(i,i)\mid i\in\eint{1}{\min \{x,y\}}\},& \text{ otherwise}.\\
\end{cases}$$
Here, the moves that can be applied from a position $(x,y)$ depend on the position itself.

Recently, Fraenkel and Larsson introduced a generalization of this notion of invariance \cite{FraLar}.

\begin{definition}
Let $t\ge 1$ be an integer. A game $G:\mathbb{N}^n\to 2^{\mathbb{N}^n}$ is {\em $t$-invariant} if the set of positions can be partitioned into $t$ subsets $S_1,\ldots,S_t$ and there exist $t$ sets $I_1,\ldots,I_t\subseteq\mathbb{N}^n$ such that, for all positions $\mathbf{p}$,
$$\text{if }\mathbf{p}\in S_j\text{, then }G(\mathbf{p})=I_j\cap \{\mathbf{m}\in\mathbb{N}^n\mid \mathbf{m}\le \mathbf{p}\}.$$
\end{definition}

In particular, an invariant game is $1$-invariant.
\begin{example}
The game $G_{\mathrm{EVEN}}$ is clearly $2$-invariant. One considers the partition of $\mathbb{N}^2$ into $S_1=\{(x,y)\mid x+y\text{ is even}\}$ and $S_2=\{(x,y)\mid x+y\text{ is odd}\}$.
\end{example}

Note that there exist some games which are not $t$-invariant for any $t$.
\begin{example}
The game
$$G_{\mathrm{MARK}}:\mathbb{N}\to 2^{\mathbb{N}}, x\mapsto
    \left\{1,\lceil x/2\rceil\right\}
$$
defined in \cite{mark} is not $t$-invariant for any $t$.
\end{example}

It is classical to associate a set of $\mathcal{P}$-positions with a game.

\begin{definition}
    A position $\mathbf{p}\in\mathbb{N}^n$ is a {\em $\mathcal{P}$-position} if there exists a strategy for the second player
(i.e., the player who will play on the next round) to win the
game, whatever the move of the first player is. We let $\mathcal{P}(G)$ denote the set of $\mathcal{P}$-positions of the game $G$. Conversely, $\mathbf{p}$ is an
{\em $\mathcal{N}$-position} if there exists a winning strategy for the
first player (i.e., the one who is making the current move).
\end{definition}

The characterization of the set of $\mathcal{P}$-positions
of an impartial acyclic game is well-known.

\begin{proposition}\label{pro:noyau} The sets of $\mathcal{P}$- and $\mathcal{N}$-positions
of an
    impartial acyclic game are uniquely determined
by the
    following two properties:
\begin{itemize}
  \item Every move from a $\mathcal{P}$-position leads to an $\mathcal{N}$-position
    (stability property of the set of $\mathcal{P}$-positions).
  \item From every $\mathcal{N}$-position, there exists a move leading to a $\mathcal{P}$-position (absorbing property of the set of $\mathcal{P}$-positions).
\end{itemize}
\end{proposition}

\begin{remark}
    Two different games $G$ and $H$ can be such that $\mathcal{P}(G)=\mathcal{P}(H)$. For an example, the game
$$H:\mathbb{N}^2\to 2^{\mathbb{N}^2},$$ 
$$(x,y)\mapsto G_{\mathrm{WYTHOFF}}(x,y) \cup\bigl(\{(3,1),(5,1),(6,1)\}\cap\{(i,j)\mid i\le x, j\le y\}\bigr)$$
is such that $\mathcal{P}(G_{\mathrm{WYTHOFF}})=\mathcal{P}(H)$. See \cite{FraNow} for details about the moves that can be adjoined to Wythoff's game without modifying the set of $\mathcal{P}$-positions. Here $G_{\mathrm{WYTHOFF}}$ and $H$ are both invariant games, but one can also imagine variant games leading to the same set of $\mathcal{P}$-positions such as
$$      H':\mathbb{N}^2\to 2^{\mathbb{N}^2},$$ 
$$(x,y)\mapsto
        \begin{cases}
G_{\mathrm{WYTHOFF}}(x,y) \cup\left(\{(5,1),(6,1)\}\cap\{(i,j)\mid i\le x, j\le y\}\right), \text{ if }x\text{ even};\\
G_{\mathrm{WYTHOFF}}(x,y) \cup\left(\{(1,3),(1,5)\}\cap\{(i,j)\mid i\le x, j\le y\}\right), \text{ if }x\text{ odd}.\\
        \end{cases}$$
Again, one can check that $\mathcal{P}(G_{\mathrm{WYTHOFF}})=\mathcal{P}(H')$. A few extra moves are adjoined to the usual Wythoff moves. These moves do not modify the set of $\mathcal{P}$-positions. But observe that the added moves depend on the parity of $x$, thus this game is not invariant. For other variant games having $\mathcal{P}(G_{\mathrm{WYTHOFF}})$ as set of $\mathcal{P}$-positions, see \cite{Bo}. Observe that $H'$ is an example of $2$-invariant game.
\end{remark}

This leads to the following definition. Note that we therefore have two notions of invariance: one for games and one for sets.

\begin{definition}
    A subset $X$ of $\mathbb{N}^n$ is {\em $t$-invariant} if there exists a $t$-invariant game $G$ such that $X=\mathcal{P}(G)$.
\end{definition}

In this paper, we deal with the question of {\em $t$-invariance} of subsets of $\mathbb{N}^n$. We first show that every subset of $\mathbb{N}^n$, which contains $\mathbf{0}=(0,\ldots,0)$, is $2$-invariant (Theorem~\ref{the:inv_ss}). Thus the general question addressed in this paper is the following one. Given a subset $X$ of $\mathbb{N}^n$, is $X$ $1$-invariant ? In \cite{Lar}, the authors proved that if $X$ is a pair of complementary homogeneous Beatty sequences, then $X$ is $1$-invariant. Recently, the case of non-homogeneous Beatty sequences was investigated \cite{CaDR}: a partial characterization is given for sets which are $1$-invariant. In the current paper, using the formalism of first-order logic, we show that for a wide range of sets (not only Beatty sequences), the problem turns out to be decidable. More precisely, this problem is decidable for sets that are definable in the Presburger arithmetic $\langle\mathbb{N},+\rangle$ extended with a unary map $V_U$ related to expansions of integers in a numeration system. The precise framework is given in Section~\ref{sec:3}. As a particular case, we will also consider sets of $\mathcal{P}$-positions of existing variant games, such as the rat game and the mouse game \cite{Rat}, the Tribonacci game \cite{tribo}, Pisot cubic games \cite{pisot}, Mark \cite{mark}, etc.


As an example, consider the Tribonacci game $G_{\mathrm{TRIBO}}:\mathbb{N}^3\to 2^{\mathbb{N}^3}$, played over three piles of tokens, where the variant rules are described in \cite{tribo}. Its set of $\mathcal{P}$-positions is coded by the Tribonacci word $T$
\begin{equation}
    \label{eq:tw}
    T=12131211213121213121121312131211213121213121\cdots,
\end{equation}
which is the unique fixed point of the morphism over $\{1,2,3\}^*$ given by $f:1\mapsto 12,2\mapsto 13,3\mapsto 1$. If letters of $T$ are indexed by positive integers, the $m$th $\mathcal{P}$-position is given by the index of the $m$th occurrence of the letter $1$, $2$ and $3$ respectively. The first few $\mathcal{P}$-positions in $\mathcal{P}(G_{\mathrm{TRIBO}})$ are
$$(1, 2, 4), (3, 6, 11), (5, 9, 17), (7, 13, 24), (8, 15, 28), (10, 19, 35), (12, 22, 41),\ldots.$$
Up to implementation, our main result shows that it is decidable whether $\mathcal{P}(G_{\mathrm{TRIBO}})$ is $1$-invariant.  For an introduction to combinatorics on (infinite) words, see, for instance, \cite{cant,Lot}. In this paper, we only use a few properties and definitions about words. For the Tribonacci word, one has to consider the sequence of finite words $(f^n(1))_{n\ge 1}$ which converges to $T$. We also assume that the reader has some knowledge about automata theory. See for instance \cite{saka}. Indeed, in Section~\ref{sec:recsyn}, we recall that the first-order logical formalism is equivalent to a representation in terms of languages accepted by finite automata.

This paper is organized as follows. In Section~\ref{sec:2}, we quickly show that every subset of $\mathbb{N}^n$ containing $\mathbf{0}=(0,\ldots,0)$ is $2$-invariant.  In Section~\ref{sec:3}, we describe the formalism of Pisot numeration systems and the corresponding first-order definable sets. We explain in Section~\ref{sec:4} how this formalism leads to a decision procedure about the $1$- or $2$-invariance of subsets of $\mathbb{N}^n$. Next, in Section~\ref{sec:recsyn}, we reformulate our result in terms of sets recognizable by means of finite automata.
Section~\ref{sec:5} is dedicated to applications of this procedure. We consider several classical games. We can decide the $1$-invariance of the set of $\mathcal{P}$-positions of $2$-heap games defined in \cite{Fra:98}, (generalized) Tribonacci game \cite{tribo,pisot}, Raleigh game \cite{Raleigh}, games coded by periodic words or by Parikh-constant morphic words. Note that in the periodic case, we also provide an independent algorithm.

\section{Invariant set of positions}\label{sec:2}

If $A,B$ are subsets of $\mathbb{N}^n$, $A+B$ (resp., $A-B$) denotes the set $\{\mathbf{a}+\mathbf{b}\mid \mathbf{a}\in A,\mathbf{b}\in B\}$ (resp., $\{\mathbf{a}-\mathbf{b}\mid \mathbf{a}\in A,\mathbf{b}\in B\}$).

\begin{theorem}\label{the:inv_ss}
    Every subset $P$ of $\mathbb{N}^n$ which contains $\mathbf{0}=(0,\ldots,0)$ is $2$-invariant.
\end{theorem}

\begin{proof}
  We define a partition of $\mathbb{N}^n$ into two subsets $S_1$ and $S_2$ as follows. First we make use of an auxiliary set $S$.
Let $\mathbf{x}$ be not in $P$. If there exists $\mathbf{y}\in P$ such that $\mathbf{y}\leq\mathbf{x}$ and $\mathbf{x}-\mathbf{y}\not\in P-P$, then $\mathbf{x}$ belongs to $S$. Otherwise, if, for all $\mathbf{y}$ in $P$ such that  $\mathbf{y}\leq\mathbf{x}$, we have $\mathbf{x}-\mathbf{y}\in P-P$ then $\mathbf{x}$ belongs to $S_2$. To get a partition of $\mathbb{N}^n$, we define $S_1$ as $S\cup P$.

We define $I_1=\mathbb{N}^n\setminus (P-P)$ and $I_2=S_2$. The game $G:\mathbb{N}^n\to 2^{\mathbb{N}^n}$ defined, for every $\mathbf{p}\in\mathbb{N}^n$, by $$\text{if }\mathbf{p}\in S_j\text{, then }G(\mathbf{p})=I_j\cap \{\mathbf{m} \in\mathbb{N}^n \mid \mathbf{m}\le \mathbf{p}\}, \text{ for }j=1,2,$$ is $2$-invariant.
Clearly, this game has $P$ as set of $\mathcal{P}$-positions, i.e., $\mathcal{P}(G)=P$. Indeed, with this definition, a player in a position $\mathbf{p}\in P$ can only play to a position not in $P$. Now assume that we have a position $\mathbf{p}\not\in P$. If $\mathbf{p}$ is in $S_1$ then, by definition of $S$, there is an option of $\mathbf{p}$ in $P$. If $\mathbf{p}$ is in $S_2$, then the move $\mathbf{p}$ leading directly to $\mathbf{0}\in P$ is allowed.
\end{proof}

\section{A first-order logic formalism}\label{sec:3}

We start with a minimal background on (Pisot) numeration systems, and then define what is a $U$-definable subset of $\mathbb{N}^n$. We conclude this section by the statement of B\"uchi's theorem which is at the center of our decision procedure.

A real number $\alpha>1$ is a {\em Pisot number} if it is the root of polynomial $P$ over $\mathbb Z$  whose dominant coefficient is $1$ ($\alpha$ is an algebraic integer) and whose all the others roots have modulus less than one.


\begin{definition}\cite{BH}
    A {\em Pisot numeration system} is an increasing sequence $(U_i)_{i\ge 0}$ of integers such that $U_0=1$ and $(U_i)_{i\ge 0}$ satisfies a linear recurrence relation whose characteristic polynomial is the minimal polynomial of a Pisot number $\alpha$.
\end{definition}

If $U=(U_i)_{i\ge 0}$ is a Pisot numeration system, every non-negative integer $m$ has a unique $U$-expansion denoted by $\rep_U(m)=c_\ell\cdots c_0$ computed by a greedy algorithm \cite{Fra:85}. It satisfies, for all $r\le\ell$,
$$\sum_{j=0}^r c_j\, U_j<U_{r+1}$$
and $c_\ell\neq 0$.
The Pisot condition implies that there exists $d>0$ such that $\lim_{i\to+\infty} U_i/\alpha^i=d$.
Hence $U_{i+1}/U_i$ converges to $\alpha$, and the digit-set for $U$-expansions is finite and equal to $\eint{0}{\sup \lceil U_{i+1}/U_i \rceil-1}$.

\begin{example}
    Let $k\ge 2$ be an integer. The sequence $(k^i)_{i\ge 0}$ is a Pisot numeration system. It is the usual base-$k$ number system. As a special case, we let $\rep_k(m)$ denote the usual base-$k$ expansion of $m$ with digit-set $\{0,\ldots,k-1\}$.
\end{example}

\begin{example}\label{exa:fibo}
    Let $(F_i)_{i\ge 0}$ be the Fibonacci sequence defined by $F_0=1$, $F_1=2$ and $F_{i+2}=F_{i+1}+F_i$ for all $i\ge 0$. The characteristic polynomial of $(F_i)_{i\ge 0}$ is $X^2-X-1$. It is the minimal polynomial of $\tau$, which is a Pisot number. The $F$-expansions are words over the digit-set $\{0,1\}$ which do not contain the factor $11$. The $F$-expansions of the first few positive integers are
$$1,10,100,101,1000,1001,1010,10000,10001,10010,10100,10101,100000,\ldots.$$
\end{example}

\begin{example}\label{exa:tribo}
    Let $(T_i)_{i\ge 0}$ be the sequence defined by $T_0=1$, $T_1=2$, $T_2=4$ and $T_{i+3}=T_{i+2}+T_{i+1}+T_i$ for all $i\ge 0$. The characteristic polynomial of $(T_i)_{i\ge 0}$ is $X^3-X^2-X-1$ which is the minimal polynomial of a Pisot number. This system is referred to as the {\em Tribonacci numeration system}. Again the digit-set is $\{0,1\}$, and $T$-expansions are words avoiding the factor $111$. The $T$-expansions of the first few positive integers are given in Table~\ref{rep:tri}
    \begin{table}[h]
        $$\begin{array}{r|ccccc|r|cccccc}
&13 & 7 &4 &2 &1 &    &24& 13 & 7 &4 &2 &1 \\
\hline
1 &  &   &  &  &1 &  14&     &1&0   &0  &0  &1 \\
2 &  &   &  & 1&0 &  15&   &1& 0  & 0 & 1&0    \\
3 &  &   &  & 1&1 &  16&  &1&   0&  0& 1&1    \\
4 &  &   & 1& 0&0 &  17&    &1&   0& 1& 0&0    \\
5 & &   & 1& 0&1 &   18&   &1&   0& 1& 0&1    \\
6 &  &   & 1& 1&0 &  19&   &1&   0& 1& 1&0    \\
7 &  & 1& 0 & 0&0 &  20&   &1& 1& 0 & 0&0     \\
8 &  & 1& 0 & 0&1 &  21&   &1& 1& 0 & 0&1    \\
9 &  & 1& 0 & 1&0 &  22&  &1& 1& 0 & 1&0   \\
10&   & 1& 0 & 1&1 &  23&   &1& 1& 0 & 1&1   \\
11&   & 1& 1 & 0&0 &  24&  1& 0& 0&  0& 0& 0  \\
12&   & 1& 1 & 0&1 &  25&  1& 0& 0&  0& 0& 1 \\
13& 1 & 0& 0 & 0&0 &  26&  1& 0& 0&  0& 1&0\\
\end{array}$$
        \caption{The $T$-representations of the first few integers.}
        \label{rep:tri}
    \end{table}
\end{example}

We recall basic definition about formal logic. See, for instance, \cite{BH} or \cite[Chapter~6]{flans}. Let $U$ be a Pisot numeration system. The alphabet $A_U$ of our first-order language contains countably many variables $x_1,x_2,x_3,\ldots$ (or $x,y,z,\ldots$) and extra symbols $+,=,V_U,\vee,\wedge,\to,\leftrightarrow,\neg,\exists,\forall$ as well as parentheses. We define terms and formulae inductively and we often make use of extra parentheses for the sake of clarity of the constructions.

\begin{definition}
First we define {\em terms} which are particular words over $A_U^*$. To construct them, we may apply the following rules finitely many times.
\begin{itemize}
  \item Any variable is a term.
  \item If $t$ is a term, then $V_U(t)$ is a term.
  \item If $t_1,t_2$ are terms, then $(t_1+t_2)$ is a term.
\end{itemize}
Second we define {\em formulae} inductively by applying the following rules finitely many times. Note that these formulae are also words over $A_U^*$.
\begin{itemize}
  \item If $t_1,t_2$ are terms, then $t_1=t_2$ is a formula.
  \item If $\varphi$ and $\psi$ are formulae, then
$\varphi\vee\psi,\ \varphi\wedge\psi,\ \varphi\to\psi,\ \varphi\leftrightarrow\psi,\ \neg\varphi$
are formulae.
\item If $\varphi$ is a formula and $x$ is a variable, then $\forall x\varphi$ and $\exists x\varphi$ are formulae.
\end{itemize}
The set of formulae over $A_U$ is the {\em first-order language} that we shall consider. We denote by $\mathfrak{n}_U=\langle \mathbb{N},+,V_U\rangle$ this structure.
\end{definition}

We turn to the interpretation (i.e., semantics) of these formulae. Variables are ranging over $\mathbb{N}$. Note that we can only quantify over variables (in contrast with second-order logic). The function $V_U$ maps $m>0$ to the smallest $U_i$ appearing in the $U$-expansion of $m$ with a non-zero coefficient. We set $V_U(0)=1$. In the special case of an integer base system, $V_k(m)$ is the largest power of $k$ dividing $m$. The other symbols carry their usual interpretation.

Let $\varphi$ be a {\em sentence}, i.e., a formula with no free variable (all the variables are under the scope of a quantifier). We write $\mathfrak{n}_U\models \varphi$ if the formula $\varphi$ is satisfied under the usual interpretation of the symbols. As an example, we write
$$\mathfrak{n}_U\models (\forall x)(\exists y)(x=y+y \vee x=y+y+1)$$
to express that every non-negative integer $x$ is either even or odd (adding one is permitted, see Remark~\ref{rem:var}). The set of sentences is called the {\em first-order theory} (of the corresponding language).

Now assume that $\varphi$ is a formula where $n$ free variables $x_1,\ldots,x_n$ occur. We write $\varphi(x_1,\ldots,x_n)$ to highlight the presence of these free variables. If we substitute all the occurrences of $x_1,\ldots,x_n$ with constants $d_1,\ldots,d_n\in\mathbb{N}$ respectively, then the resulting formula is a sentence and either $\mathfrak{n}_U\models \varphi(d_1,\ldots,d_n)$ or not. Given a formula $\varphi(x_1,\ldots,x_n)$, we can therefore consider the $n$-tuples $(d_1,\ldots,d_n)\in\mathbb{N}^n$ such that $\mathfrak{n}_U\models \varphi(d_1,\ldots,d_n)$. This leads to the following definition.

\begin{definition}
Let $n\ge 1$. Let  $U=(U_i)_{i\ge 0}$ be a Pisot numeration system.
    A set $X\subseteq\mathbb{N}^n$ is {\em $U$-definable} if there exists a formula $\varphi(x_1,\ldots,x_n)$ in $\langle \mathbb{N},+,V_U\rangle$ such that
$$X=\{(d_1,\ldots,d_n)\in\mathbb{N}^n\mid \mathfrak{n}_U\models \varphi(d_1,\ldots,d_n)\}.$$
If $U$ is the usual base-$k$ number system, we speak of {\em $k$-definable} sets.
\end{definition}

\begin{remark}\label{rem:var}
We can define integer constants, inequality relation, multiplication by a constant, and Euclidean division by a constant within $\langle \mathbb{N},+,V_U\rangle$. The order relation $x\le y$ is defined by $(\exists z)(y=x+z)$. The constant $0$ is defined by $x=0\equiv (\forall y)(x\le y)$. The successor function $S$ mapping $n$ to $n+1$ is defined by $S(x)=y\equiv (x\le y)\wedge (\forall z)((\neg(x=z)\wedge(x\le z)) \to (y\le z))$. To define $r\cdot x$, where $r$ is a constant, one has to write $x+\cdots +x$, where the sum has $r$ terms.
\end{remark}

\begin{example}\label{exa:phi}
    Consider the formula $$\psi(x)\equiv(\exists y)(x=y+y).$$ The formula $\psi$ defines the set of even integers (whatever is the Pisot numeration system $U$). Let $k\ge 2$. The formula $$\nu(x)\equiv V_k(x)=x$$ defines the set of powers of $k$. We can say that this set is $k$-definable. Let $(T_i)_{i\ge 0}$ be the Tribonacci sequence from Example~\ref{exa:tribo}. The formula $$\phi(n)\equiv (\exists x)(\exists y)(x\neq y\wedge V_T(x)=x\wedge V_T(y)=y\wedge n=x+y)$$ defines the set of number whose $T$-expansion contains exactly two non-zero terms. The first few integers in the set defined by $\phi$ are $3,5,6,8,9,11,14,15$ because their $T$-expansions contain exactly two symbols $1$.
\end{example}

Given any sentence $\varphi$ in $\langle \mathbb{N},+,V_U\rangle$, there exists an algorithm to check whether $\mathfrak{n}_U\models \varphi$ holds. This result is known as B{\"u}chi's theorem. See \cite{Buc,BHMV}.

\begin{theorem}\label{the:buchi}
    The first-order theory $\langle \mathbb{N},+,V_U\rangle$ is effectively decidable.
\end{theorem}

\begin{remark}\label{imp}
    This result has been recently used to get positive results in combinatorics on words. Implementations to deal with the Fibonacci and Tribonacci numerations systems have been developped \cite{aut1,aut2}. With these implementations (mostly relying on automata recognizing addition in these systems) many properties of the Fibonacci, Tribonacci and related infinite words are proved automatically on a laptop with computing time ranging from a few seconds to two hours. The source code developped by the authors of \cite{aut1,aut2} has not yet been publicly released. Also see \cite{goc2} for an example about integer base systems.
\end{remark}

\section{Decision procedure}\label{sec:4}

Let us make a preliminary observation. A move can be adjoined to an impartial acyclic  game without
changing the set of $\mathcal{P}$-positions if and only if it does not alter the stability property (defined in
Proposition~\ref{pro:noyau}). Indeed, adding a move leading from
one $\mathcal{P}$-position to another $\mathcal{P}$-position would necessarily
change the stability property of the $\mathcal{P}$-positions (by
Proposition~\ref{pro:noyau}). On the other hand, adding a move
which does not correspond to a move between any two $\mathcal{P}$-positions means that both properties of Proposition~\ref{pro:noyau} remain true. Therefore, a move $\mathbf{m}$ can be
added if and only if it prevents a move from a $\mathcal{P}$-position to another $\mathcal{P}$-position.

\begin{lemma}\label{lem:pos}
    A subset $P$ of $\mathbb{N}^n$ is $1$-invariant if and only if the $1$-invariant game $H$ defined by $$H(\mathbf{p})=(\mathbb{N}^n\setminus (P-P))\cap \{\mathbf{m}\in\mathbb{N}^n\mid \mathbf{m}\le \mathbf{p}\}, \text{ for every }\mathbf{p},$$
has $P$ as set of $\mathcal{P}$-positions, i.e., $P=\mathcal{P}(H)$.
\end{lemma}

\begin{proof}
    Assume that there exists a $1$-invariant game $G$ such that $P=\mathcal{P}(G)$, i.e., there exists a subset $I$ such that $G(\mathbf{p})=I\cap \{\mathbf{m}\in\mathbb{N}^n\mid \mathbf{m}\le \mathbf{p}\}$ for every $\mathbf{p}$. For all positions $\mathbf{p}$, we have
$$I\subseteq(\mathbb{N}^n\setminus (P-P))\text{, and thus } G(\mathbf{p})\subseteq H(\mathbf{p}),$$
because if a move $\mathbf{m}$ belongs to $I$, then $\mathbf{m}$ cannot be written as $\mathbf{x}-\mathbf{y}$ for two distinct elements $\mathbf{x},\mathbf{y}\in P$. Indeed if that were the case, $\mathbf{m}=\mathbf{x}-\mathbf{y}$ would belong to $G(\mathbf{x})$, and we would be able to play from $\mathbf{x}\in P$ to another element $\mathbf{y}\in P$. This contradicts the fact that $P$ is the set of $\mathcal{P}$-positions of $G$.

If we compare the two games, $H$ is an extension of $G$: for every position, we could have more options in $H$ than in $G$. Nevertheless, for a position in $P$ if there are more options available in $H$, the new options do not belong to $P$.
Using Proposition~\ref{pro:noyau}, we deduce that $P=\mathcal{P}(G)=\mathcal{P}(H)$.
\end{proof}

\begin{theorem}\label{the:main}
Let $n\ge 1$. Let $U$ be a Pisot numeration system.
    Let $P$ be a $U$-definable subset of $\mathbb{N}^n$ containing $\mathbf{0}$. It is decidable whether $P$ is an $1$-invariant set.
\end{theorem}

\begin{proof}
    Without loss of generality, we may assume that $n=2$. Let $\pi(x,y)$ be a first-order formula defining $P$.
The set $(P-P)\cap\mathbb{N}^2$ is $U$-definable by the following formula
$$\varphi(x,y)\equiv(\exists i_1)(\exists i_2)(\exists j_1)(\exists j_2)(\pi(i_1,i_2)\wedge \pi(j_1,j_2)\wedge j_1\ge i_1\wedge j_2\ge i_2\wedge x=j_1-i_1\wedge y=j_2-i_2).$$
From the above lemma, the fact that $P$ is $1$-invariant can be expressed by the validity of the sentence
$$\nu\equiv (\forall p_1)(\forall p_2) (\neg\pi(p_1,p_2)\to (\exists x)(\exists y)(\pi(x,y)\wedge x\le p_1\wedge y\le p_2\wedge \neg\varphi(p_1-x,p_2-y))).$$
Indeed, for every position $(p_1,p_2)$ which is not in $P$, we are looking for $(x,y)\in P$ and a move not in $P-P$ from  $(p_1,p_2)$ to $(x,y)$. If this holds for every $(p_1,p_2)$ not in $P$, then the invariant game defined in Lemma~\ref{lem:pos} has $P$ as set of $\mathcal{P}$-position.
We finish the proof by an application of Theorem~\ref{the:buchi}. We can decide if the sentence $\nu$ holds.
\end{proof}

\section{Recognizable sets and synchronized sequences}\label{sec:recsyn}

A key ingredient for the proof of Theorem~\ref{the:buchi} comes from automata theory. B{\"u}chi's proof is constructive: a finite automaton is associated with every formula and conversely \cite{Buc}. We will thus reformulate Theorem~\ref{the:main} in terms of regular languages, i.e., sets of $n$-tuples of words recognized by finite automata. Again for details, see \cite{BH,BHMV}.

\begin{definition}
Let $n\ge 1$. Let  $U=(U_i)_{i\ge 0}$ be a Pisot numeration system. We first define the $U$-expansion of an $n$-tuple of integers as
$$\rep_U(x_1,\ldots,x_n)=(0^{\ell-|\rep_U(x_1)|}\rep_U(x_1),\ldots,0^{\ell-|\rep_U(x_n)|}\rep_U(x_n))$$
where $\ell=\max_{j\in\{1,\ldots,n\}}|\rep_U(x_j)|$.\end{definition}

\begin{example}
    Consider the sequence $(F_i)_{i\ge 0}$ from Example~\ref{exa:fibo}. For the pair $(6,10)$, we have $\rep_F(6)=1001$, $\rep_F(10)=10010$, and thus $\rep_F(6,10)=(01001,10010)$.
\end{example}

 The idea is to have $n$ components of the same length; so the shortest expansions are padded to the left with zeroes. Hence an automaton can read simultaneously the $i$th digit of every component.

\begin{definition}
    A set $X\subseteq\mathbb{N}^n$ is {\em $U$-recognizable} if there exists a deterministic finite automaton (DFA for short) reading $n$-tuples of digits and recognizing the set
$$\rep_U(X)=\{\rep_U(x_1,\ldots,x_n)\mid (x_1,\ldots,x_n)\in X\}.$$
If $U$ is the usual base-$k$ number system, we speak of {\em $k$-recognizable} sets.
\end{definition}

We use standard conventions to represent DFA. The initial state has an incoming arrow. The final states have an outgoing arrow. Words are read from left to right, i.e., most significant digit first.
If a transition is not depicted, then it leads to a dead state. Note that we allow leading zeroes (which does not affect recognizability by finite automaton).

\begin{example} The DFA depicted in Figure~\ref{fig:V2a} recognizes the base-$2$ expansions of pairs $(x,y)$ such that $y$ is the largest power of $2$ dividing $x$.  For instance, $(1,1)$, $(101010,000010)$ or $(11100,00100)$ belong to the recognized language. Otherwise stated, the set $\{(n,V_2(n))\mid n>0\}$ is $2$-recognizable.
    \begin{figure}[htbp]
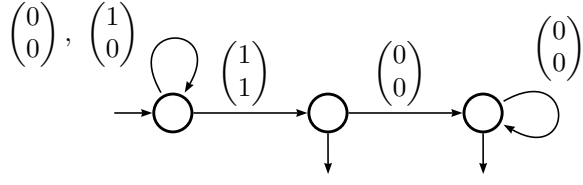

        \centering
\VCDraw{%
        \begin{VCPicture}{(-2,-2)(6,2.5)}
 \StateVar[]{(-1.7,0)}{3}
 \StateVar[]{(1.7,0)}{4}
 \StateVar[]{(5.1,0)}{5}
\Initial[w]{3}
\Final[s]{4}
\Final[s]{5}
\LoopN{3}{\begin{pmatrix}
    0\\0\\
\end{pmatrix},\
\begin{pmatrix}
    1\\0\\
\end{pmatrix}}
\LoopE{5}{\begin{pmatrix}
    0\\0\\
\end{pmatrix}}
\EdgeL{3}{4}{\begin{pmatrix}
    1\\1\\
\end{pmatrix}}
\EdgeL{4}{5}{\begin{pmatrix}
    0\\0\\
\end{pmatrix}}
\end{VCPicture}
}
\caption{A DFA recognizing a language over $\{0,1\}^2$.}
        \label{fig:V2a}
    \end{figure}
\end{example}

\begin{theorem}\cite{BH}\label{the:equiv} Let $n\ge 1$. Let  $U=(U_i)_{i\ge 0}$ be a Pisot numeration system.
    A set $X\subseteq\mathbb{N}^n$ is $U$-recognizable if and only it is $U$-definable.
\end{theorem}

\begin{example}
    The DFA depicted in Figure~\ref{fig1} recognizes the $T$-expansions containing exactly two symbols $1$. It means that the corresponding set of integers is $T$-recognizable. Recall that this set also is $T$-definable using the formula $\psi$ of Example~\ref{exa:phi}.
\begin{figure}[htbp]
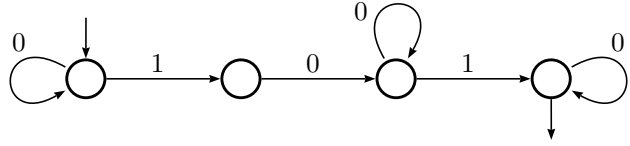

        \centering
\VCDraw{%
        \begin{VCPicture}{(-2,-1)(6,2)}
 \StateVar[]{(-1.7,0)}{1}
 \StateVar[]{(1.7,0)}{2}
 \StateVar[]{(5.1,0)}{3}
 \StateVar[]{(8.5,0)}{4}
\Initial[n]{1}
\Final[s]{4}
\LoopW{1}{0}
\LoopE{4}{0}
\LoopN{3}{0}
\EdgeL{1}{2}{1}
\EdgeL{2}{3}{0}
\EdgeL{3}{4}{1}
\end{VCPicture}
}
\caption{A DFA recognizing a $T$-definable set.}
        \label{fig1}
    \end{figure}
\end{example}

Thus our main theorem (Theorem~\ref{the:main}) can also be restated as follows. We state this result for two sequences, but there is no problem to extend this result to $n$-tuples.

\begin{theorem}\label{thm:main}
    Let $U$ be a Pisot numeration system. If a set $P$ is $U$-recognizable, then it is decidable whether $P$ is a $1$-invariant set.
\end{theorem}

The notion of $k$-synchronized sequences is classical \cite{Carpi}. The idea is that the graph of the function, i.e., the set of pairs $(n,f(n))_{n\ge 0}$, is $k$-recognizable. As we will easily see, synchronized sequences are sufficient to apply our theorem.
\begin{definition}
Let  $U=(U_i)_{i\ge 0}$ be a Pisot numeration system. A sequence $(x_i)_{i\ge 0}$ of non-negative integers is {\em $U$-synchronized} if the set $\{(i,x_i)\mid i\ge 0\}\subset\mathbb{N}^2$ is $U$-recognizable.
\end{definition}

\begin{example} The DFA depicted in Figure~\ref{fig2} recognizes exactly the set $\{(i,2i)\mid i\ge 0\}$. Otherwise stated, the sequence $(2i)_{i\ge 0}$ is $2$-synchronized.
    \begin{figure}[htbp]
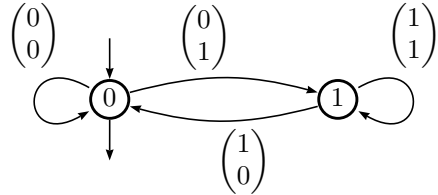

        \centering
\VCDraw{%
        \begin{VCPicture}{(-2,-2)(3,2)}
 \StateVar[0]{(-2,0)}{1}
 \StateVar[1]{(3,0)}{2}
\Initial[n]{1}
\Final[s]{1}
\LoopW{1}{
  \begin{pmatrix}
      0\\0\\
  \end{pmatrix}
}
\LoopE{2}{
  \begin{pmatrix}
      1\\ 1\\
  \end{pmatrix}
}
\ArcL{1}{2}{
  \begin{pmatrix}
      0\\ 1\\
  \end{pmatrix}
}
\ArcL{2}{1}{
  \begin{pmatrix}
      1\\ 0\\
  \end{pmatrix}
}
\end{VCPicture}
}
\caption{A DFA recognizing a $2$-synchronized relation.}
        \label{fig2}
    \end{figure}
\end{example}

\begin{corollary}
     Let $n\ge 1$. Let  $U=(U_i)_{i\ge 0}$ be a Pisot numeration system. Let $(A_i)_{i\ge 0}$ and $(B_i)_{i\ge 0}$ be two $U$-synchronized sequences.  Then the set $P=\{(A_k,B_k)\mid k\ge 0\}$ is $U$-recognizable, and thus it is decidable whether $P$ is a $1$-invariant set.
\end{corollary}

\begin{proof}
    The set of regular languages is closed under intersection and projection. The languages $\{\rep_U(i,A_i)\mid i\ge 0\}$ and $\{\rep_U(i,B_i)\mid i\ge 0\}$ are regular. Hence, one easily derives that the language $\{\rep_U(A_i,B_i)\mid i\ge 0\}$ is also regular. Indeed, the languages $\{\rep_U(A_i,i,j)\mid i,j\ge 0\}$ and $\{\rep_U(j,i,B_i)\mid i,j\ge 0\}$ are regular. The intersection of these two sets is the regular language $\{\rep_U(A_i,i,B_i)\mid i\ge 0\}$, and thus the projection on the first and third component also is.
\end{proof}

\section{Applications}\label{sec:5}

Let $\mathcal{A}=\{1,\ldots,n\}$ be a finite alphabet. An infinite word $w\in\mathcal{A}^{\mathbb{N}_{>0}}$ represents a subset of $\mathbb{N}^n$ as follows. It is usual to assume that each symbol occurs infinitely often in $w$. We consider the $n$-tuple $(i_{m,1},\ldots,i_{m,n})$, where $i_{m,j}$ denotes the index of the $m$th symbol $j$ occurring in $w$ (as we did in the introduction for the Tribonacci game). Recall that the first letter in $w$ has index $1$.
\begin{definition}\label{def:pw}
We let $P_w$ denote the subset of $\mathbb{N}^n$ which is made up of all these $n$-tuples described above and also their permutations, and we add the $n$-tuple $\mathbf{0}$.
\end{definition}

Therefore an infinite word $w$ over $\mathcal{A}=\{1,\ldots,n\}$ can be a convenient alternative to represent or characterize a set of $\mathcal{P}$-positions for a take-away game over $n$ piles. Particularly when the infinite word is generated by a simple procedure such as iterating a morphism defined over a finite alphabet. For an example, see \eqref{eq:tw}. Games such as Wythoff's game or the Tribonacci game have a set of $\mathcal{P}$-positions coded by a morphic word.

\begin{example}
    The infinite Fibonacci word $F=12112121121121211212112\cdots$ is a fixed point of the morphism $1\mapsto 12, 2\mapsto1$. It defines a subset $P_F$ of $\mathbb{N}^2$. The first few elements in $P_F$ are
$$(0,0), (1,2), (2,1), (3,5), (5,3), (4,7), (7,4), (6,10), (10,6), (8,13), (13,8), (9,15), (15,9), \ldots$$
It is well known that the $n$th $\mathcal{P}$-position of Wythoff's game is given by the position of the $n$th symbol $1$ and $n$th symbol $2$ occurring in $F$.
\end{example}

\begin{remark}\label{rem:AmBm}
    Note that the use of infinite words also ensures that we get a partition of $\mathbb{N}_{>0}$. Assume that we have a binary alphabet $\{1,2\}$. Let $A_m$ (resp., $B_m$) be the position of the $m$th letter $1$ (resp., $2$) in an infinite word $w$. Then $\{A_m\mid m\ge 1\}\cup \{B_m\mid m\ge 1\}=\mathbb{N}_{>0}$ and $\{A_m\mid m\ge 1\}\cap \{B_m\mid m\ge 1\}=\emptyset$. In addition, we also consider the permutations of the $n$-tuples in Definition~\ref{def:pw} since it provides "symmetric" rulesets, as it is the case in many take-away games.
\end{remark}

In what follows, we will give examples of sets $P_w$ which satisfy the condition of Theorem~\ref{the:main}, i.e., sets for which we can decide whether they are $1$-invariant or not.

\subsection{Periodic words}

If we analyze the implementation of the decision procedure derived from Theorem~\ref{the:main}, it is well known that the running time is bounded by an expression of the form
$$2^{2^{\cdot^{\cdot^{\cdot^{2^{p(n)}}}}}}$$
where $p$ is a polynomial. The height of tower of exponents corresponds to the number of nested quantifiers in the used logical formula (a quantifier leads to the construction of a non-deterministic automaton, and the subset algorithm for determinization explains this possible blow-up). One can doubt that such a bad complexity could lead to effective results.  Yet, even though we have a very bad worst case scenario, positive results have been obtained in the field of combinatorics on words. See, for instance, \cite{goc} and Remark~\ref{imp}. It motivated us to obtain a different decision procedure (not based on the first order logic arithmetic) in some particular cases. Here we consider the special case of sets $P_w$ defined by a binary periodic word $w=v^\omega=vvv\cdots$. Note that our result should remain true for words over a larger alphabet.

We explain in details the case where $v$ is a word on $\{1,2\}$ that contains as many ones than twos. Let $\ell=|v|$ denote the length of $v$. We let $|v|_a$ denote the number of occurrences of the symbol $a$ within $v$. Hence we assume that $|v|_1=|v|_2=\ell/2$.
The basic idea is that the periodicity allows us to check the $1$-invariance only on the small square $\eint{0}{\ell}^2$. The set $P_w$ is periodic in the following sense. Let $\mathbf{p}\in P_w$. If $\mathbf{p}\neq {(0,0)}$, then $\mathbf{p}+(\ell,\ell)\in P_w$.
Indeed, if $\mathbf{p}=(A_m,B_m)$ (with the notations of Remark~\ref{rem:AmBm}), then $\mathbf{p}+(\ell,\ell)=(A_{m+\ell/2},B_{m+\ell/2})$. Similarly, if $\mathbf{p}>(\ell,\ell)$, then $\mathbf{p}-(\ell,\ell)\in P_w$. In other words, if $P_w^0$ denotes the set $P_w\cap \eint{1}{\ell}^2$, then $$P_w=\{(0,0)\}\cup (P_w^0+ \{k\cdot(\ell,\ell) \ | \ k \in \mathbb N\}).$$
Furthermore, let $s$ be the unique integer such that $s \frac{\ell}{2}<m\leq (s+1)\frac{\ell}{2}$. Then $\mathbf{p}=(A_m,B_m)\in \eint{s \ell+1}{(s+1)\ell}^2$. Indeed, the $m$th occurrence of $1$ necessarily appears in $w$ in the occurrence $s+1$ of $u$. The set $P_w-P_w$ has also some periodicity as expressed in the next lemma.

\begin{lemma}\label {lem:perioddiff}
Let $w$ be a binary periodic word on $\{1,2\}$. Let $v$ be the period of $w$ and let $\ell$ be the length of $v$.
Assume that $v$ has the same number of ones and twos.
We have $$(P_w-P_w)\cap \mathbb{N}^2=\left(P_w\cap \eint{0}{2\ell}^2 - P_w\cap \eint{0}{\ell}^2\right)\cap \mathbb{N}^2+\{k\cdot (\ell,\ell) \ | \ k\in \mathbb N\}.$$
\end{lemma}

\begin{proof}
Let $\mathbf{m} \in (P_w-P_w)\cap \mathbb{N}^2$. We can write $\mathbf{m}=\mathbf{p_1}-\mathbf{p_2}$ with $\mathbf{p_1}\geq \mathbf{p_2}$ both in $P_w$. As noticed before, $\mathbf{p_1}$ (respectively $\mathbf{p_2}$) lies in a set  $\eint{s \ell+1}{(s+1)\ell}^2$ for some $s$ (respectively in $\eint{t \ell+1}{(t+1)\ell}^2$).

Since $\mathbf{p_1}\geq \mathbf{p_2}$, we have $s\geq t$.
We can write, if $s\neq t$, $$\mathbf{m} = \left (\mathbf{p_1}-(s-1)\cdot (\ell,\ell)\right) - \left(\mathbf{p_2} - t \cdot (\ell,\ell)\right) + (s-t-1)\cdot(\ell,\ell)$$
and if $s=t$,
$$\mathbf{m} = \left(\mathbf{p_1}-s\cdot (\ell,\ell)\right) - \left(\mathbf{p_2} - s \cdot (\ell,\ell)\right).$$

Since $\mathbf{p_1}-s\cdot (\ell,\ell)$ and all the similar positions in the previous expressions are elements of $P_w$, we are done.

Assume now that $\mathbf{m}\in (P_w\cap \eint{0}{2\ell}^2 - P_w\cap \eint{0}{\ell}^2)+\{k\cdot (\ell,\ell) \ | \ k\in \mathbb N\}$.
Then $\mathbf{m}=\mathbf{p_1}-\mathbf{p_2} + k\cdot(\ell,\ell)$ for some $\mathbf{p_1}$, $\mathbf{p_2}$ and $k$.
If $\mathbf{m}\neq (0,0)$, then $\mathbf{p_1}\neq (0,0)$ and $\mathbf{p_1}+k\cdot(\ell,\ell)\in P_w$. Thus $\mathbf{m}\in P_w-P_w$.\end{proof}

The next lemma says that we only need to check the invariance on the small positions.

\begin{lemma}
Let $w$ be a binary periodic word on $\{1,2\}$. Let $v$ be the period of $w$ and let $\ell$ be the length of $v$.
Assume that $v$ has the same number of ones and twos. Then the set $P_w$ is $1$-invariant if and only if for each position $\mathbf{x}$ in $\eint{0}{2\ell}^2$ not in $P_w$, there is a position $\mathbf{p}\in P_w$ such that $\mathbf{p}\leq \mathbf{x}$ and $\mathbf{x}-\mathbf{p}\notin P_w-P_w$.
\end{lemma}

\begin{proof}
By Lemma~\ref{lem:pos}, $P_w$ is $1$-invariant if and only the game $H(P_w)$ is $1$-invariant. This is equivalent to say that for any position $\mathbf{x}\notin P_w$, there is a position $\mathbf{p}\in P_w$ such that $\mathbf{p}\leq \mathbf{x}$ and $\mathbf{x}-\mathbf{p}\notin P_w-P_w$.
Therefore assume that the latter fact is true for  $\mathbf{x}\in \eint{0}{2\ell}^2\setminus P_w$. To prove the lemma we need to prove the fact for any $\mathbf{x}\notin P_w$. We thus can assume that $\mathbf{x}\notin \eint{0}{2\ell}^2$.

Assume first that there exists an element $\mathbf{p}$ of $P_w$ such that $\mathbf{m}=\mathbf{x}-\mathbf{p}$ is equal to $(0,k)$ or $(k,0)$ for some $k$. Since the two sets $\{A_m, m\in \mathbb N_{>0}\}$ and $\{B_m, m\in \mathbb N_{>0}\}$ form a partition of $\mathbb N_{>0}$, we have $\mathbf{m}\notin P_w-P_w$, and we are done.

Assume now that this is not the case. Since the two sets $\{A_m, m\in \mathbb N_{>0}\}$ and $\{B_m, m\in \mathbb N_{>0}\}$ form a partition, we must have, without loss of generality, $\mathbf{x}=(A_m,y)$ for some $m,y$. Since $(A_m,B_m)-\mathbf{x}\neq (0,k)$ for $k\geq 0$, we  have $y<B_m$. In the same way, we also have $\mathbf{x}=(x,A_{m'})$ for some $m',x$, and $x<B_{m'}$ (it is not possible to have $y=B_{m'}$ since the positions $(A_k,B_k)$ are strictly increasing). Finally $\mathbf{x}=(A_m,A_{m'})$ with $A_m<B_{m'}$ and $A_{m'}<B_m$.
The two positions $(A_m,B_m)$ and $(B_{m'},A_{m'})$ are necessarily in the same square $\eint{s \ell+1}{(s+1)\ell}^2$ for some $s$. Hence $\mathbf{x}$ is also in this square and since $\mathbf{x}\notin \eint{0}{2\ell}^2$, we have $s\geq 2$.

We now consider the position $\mathbf{x'}=\mathbf{x}-(s-1)\cdot(\ell,\ell)$. We have $\mathbf{x'}\in \eint{0}{2\ell}^2\setminus P_w$. By hypothesis, there exists a position $\mathbf{p'}\in P_w$ such that $\mathbf{m'}=\mathbf{x'}-\mathbf{p'}\notin P_w-P_w$.
If $\mathbf{p'}\neq (0,0)$, we play the move $\mathbf{m}=\mathbf{m'}$ from $\mathbf{x}$ to the position $\mathbf{p}=\mathbf{p'}+(s-1)\cdot(\ell,\ell)$. Clearly $\mathbf{p}\in P_w$.

If $\mathbf{p'}=(0,0)$, then $\mathbf{x'}\notin P_w-P_w$, and we  will play the move $\mathbf{m}=\mathbf{x}$ to $\mathbf{p}=(0,0)$. Assume for the contradiction that $\mathbf{m}\in P_w-P_w$. Then by Lemma \ref{lem:perioddiff},  $\mathbf{x}=\mathbf{p_1}-\mathbf{p_2}+t\cdot(\ell,\ell)$ for some $\mathbf{p_1}$, $\mathbf{p_2}$ and $t$.
But we also have $\mathbf{x}=\mathbf{x'}+(s-1)\cdot(\ell,\ell)$. Since $\mathbf{p_1}-\mathbf{p_2}\in \eint{0}{2\ell}^2$ and $\mathbf{x'}\in \eint{\ell+1}{2\ell}^2$, we have $s=t$ or $s=t+1$. Then necessarily, $\mathbf{x'}=\mathbf{p_1}-\mathbf{p_2}+\delta\cdot(\ell,\ell)$ with $\delta\in \{0,1\}$. This is a contradiction since $\mathbf{x'}\notin P_w-P_w$.

\end{proof}

A natural algorithm can be deduced from the previous lemma to decide the $1$-invariance of $P_w$. Indeed, it is enough to consider the (at most) $4\ell^2$ elements $\mathbf{x}$ of $\eint{0}{2\ell}^2\setminus P_w$. For each of them, consider the at most $2\ell+1$ possible positions $\mathbf{p}\leq \mathbf{x}$ in $P_w$ and check if $\mathbf{m}=\mathbf{x}-\mathbf{p}\in P_w-P_w$. By Lemma \ref{lem:perioddiff}, since $\mathbf{m} \in \eint{0}{2\ell}^2$, if $\mathbf{m}\in P_w-P_w$, then
$$\mathbf{m}\in \left(P_w\cap \eint{0}{\ell}^2 - P_w\cap \eint{0}{\ell}^2\right)\cap \mathbb{N}^2 + \{(0,0),(\ell,\ell)\}.$$ We can compute this latter set once for all the procedure in $O(\ell^2)$ steps. Thus the total procedure runs in polynomial time in $\ell$.

\begin{corollary}
Let $w$ be a binary periodic word on $\{1,2\}$ with period $v$.
Assume that $v$ has the same number of ones and twos.
There is an algorithm running in polynomial time in $|v|$ that decides if the set $P_w$ is 1-invariant.
\end{corollary}

\begin{remark}
If $v$ does not contain the same number of ones and twos, we still have some periodicity in $P_w$ and $P_w-P_w$. Assume that $|v|_1<|v|_2$.
Let $r$ be the least common multiple of $|v|_1$ and $|v|_2$. Let $\mathbf{u}=(\frac{r}{|v|_1}\cdot |v|,\frac{r}{|v|_2}\cdot|v|)$.
Then one can check that for any $m>0$, $(A_m,B_m)+\mathbf{u}=(A_{m+r},B_{m+r})$ and thus is an element of $P_w$. Let $\mathbf{u^R}$ be the symmetric of $\mathbf{u}$ obtained by permuting its coordinates. We have:
$$P_w=\{(0,0)\}\cup \left(\{(A_m,B_m) | 1\leq m\leq r\} + \mathbb{N} \mathbf{u}\right)
\cup \left(\{(B_m,A_m) | 1\leq m\leq r\} + \mathbb{N} \mathbf{u^R}\right).$$

The set $P_w-P_w$ can be seen as the union of four different parts $D_1$, $D_2$, $D_3$ and $D_4$. The first part $D_1$ is similar to the previous case and corresponds to the difference between two elements of $P_w$ on the same form $(A_m,B_m)$ (with the convention that $(A_0,B_0)=(0,0)$):

$$D_1=\{(A_m-A_p,B_m-B_p) \ | \ 0\leq p\leq m\leq 2r \text{ and } p\leq r\} + \mathbb{N}\mathbf{u}.$$

The set $D_2$ is the symmetric of $D_1$, it corresponds to the differences between two elements $(B_m,A_m)$. The set $D_3$ corresponds to the difference between a position $(A_m,B_m)$ and a position $(B_p,A_p)$:

$$D_3=\left(\{(A_m-B_p,B_m-A_p) \ | \ 0\leq p\leq m\leq (k_0+1)r \text{ and } p\leq r\} + \{k_1\mathbf{u}-k_2\mathbf{u^R} \ | \ k_1,k_2\in \mathbb N\}\right)\cap{\mathbb N^2}$$

where $k_0$ is such that $k_0\geq\tfrac{|v|_2}{|v|_1}$.
Finally, the set $D_4$ is the symmetric of $D_3$.

With the same argument as before, it should be enough to check the $1$-invariance on the set $\eint{0}{(k_0+1)r}^2$. This would lead to a polynomial-time algorithm (in the length of $|v|$).
\end{remark}

\begin{remark}
    If $w=v^\omega$, then $P_w$ is definable in $\langle \mathbb{N},+\rangle$. So instead of our specific algorithm, one can also use the general procedure given by Theorem~\ref{the:main}.
\end{remark}

\subsection{Parikh-constant morphic words}
We turn to a slightly more general situation than $w=v^\omega$ and consider words of a special form obtained by iterating a morphism. Here, the infinite word $w$ is of the form $v_1v_2v_3\cdots$ and every finite word $v_i$ is a permutation of the letters of every $v_j$.
Let $u$ be a finite word.  A morphism $f:\mathcal{A}^*\to \mathcal{A}^*$ is {\em $\ell$-uniform} if, for all $a\in\mathcal{A}$, $|f(a)|=\ell$.
A morphism $f:\mathcal{A}^*\to \mathcal{A}^*$ is {\em Parikh-constant} if, for all $a,b,c\in\mathcal{A}$, $|f(a)|_c=|f(b)|_c$. Note that if $f$ is Parikh-constant, for all $a,b\in\mathcal{A}$, $|f(a)|=|f(b)|$, i.e., $f$ is $\ell$-uniform for some $\ell$.

\begin{lemma}\cite{BHMV}
    Let $\ell\ge 2$ be an integer. Let $w=w_1w_2w_3\cdots\in\mathcal{A}^{\mathbb{N}_{>0}}$ be a fixed point of a $\ell$-uniform morphism. For every $a\in\mathcal{A}$, the set $\{i>0\mid w_i=a\}$ is $\ell$-definable.
\end{lemma}

We let $\chi_{w,a}$ denote the formula defining the set $\{i>0\mid w_i=a\}$ given in the above lemma. Otherwise stated, $\chi_{w,a}(n)$ holds if and only if $n$ belongs to that set.

\begin{proposition}
     Let $w=w_1w_2w_3\cdots\in\mathcal{A}^{\mathbb{N}_{>0}}$ be a fixed point of a Parikh-constant morphism $f:\mathcal{A}^*\to \mathcal{A}^*$. One can decide whether $P_w$ is a $1$-invariant subset of $\mathbb{N}^{\#\mathcal{A}}$.
\end{proposition}

\begin{proof}
The morphism $f$ is $\ell$-uniform for an $\ell$.
    From Theorem~\ref{the:main}, we simply need to prove that $P_w$ is $\ell$-definable. Since $f$ is Parikh-constant, note that $\ell=|f(a)|$ and $N_b:=|f(a)|_b$, $b\in\mathcal{A}$, are given constants depending only on $f$ and not on $a$. We may assume that $N_b\ge 1$ for all $b$, otherwise we can restrict the morphism to a smaller alphabet. Let $a,b\in\mathcal{A}$ and $r\in\eint{1}{N_a}$. We let $C_{r,a,b}$ be the position of the $r$th symbol $a$ occurring in $f(b)$. Note that $C_{r,a,b}$ is a constant in $\eint{1}{\ell}$ which is also derived from $f$.

We define a predicate $g_{w,a}(m,j)$ which holds if and only if the position of the $m$th symbol $a$ occurring in $w$ is $j$,
$$g_{w,a}(m,j)\equiv (\exists q)(\exists r)( m=q\cdot N_a+r \wedge 0<r\le N_a \wedge \bigvee_{b\in\mathcal{A}}(\chi_{w,b}(q+1)\wedge j=q\cdot \ell+C_{r,a,b})).$$
Indeed, $w$ is the (infinite) concatenation of blocks of length $\ell$ and each such block contains exactly $N_a$ letters $a$. If $m=q\, N_a+r$ with $0<r\le N_a$, then the $m$th occurrence of $a$ appears in the $(q+1)$st block of length $\ell$. Since $f(w)=w$, this block is equal to $f(w_{q+1})$. The disjunction expresses the fact that $w_{q+1}=b$ for some $b$, and thus the $m$th letter $a$ occurs in position $q\, \ell+C_{r,a,b}$.
\end{proof}

\begin{example}
    Consider the morphism $f:1\mapsto 112, 2\mapsto 121$. A prefix of the fixed point $w$ of $f$ is
$$112 112 121  112 112 121  112 121 112\cdots .$$
The first few elements in $P_w$ are $(0,0), (1,3), (2,6), (4,8), (5,12), (7,15), (9,17),\ldots$. With the notation of the previous proof, $g_{w,1}(6,9)$ and $g_{w,2}(6,17)$ hold. We have $\ell=3$, $N_1=2$ and $N_2=1$. Thus $6=2 N_1+2$. This means that the sixth $1$ occurs in the third block of length $3$ which is $f(w_3)=f(2)=121$, and $C_{1,1,2}=1$ and $C_{2,1,2}=3$. Hence, the position of the sixth $1$ is $2\cdot L+C_{2,1,2}=9$.
\end{example}

\subsection{Recurrence of order $2$}

In \cite{Fra:98} Fraenkel considered a class of games whose $P$-positions can be characterized using a numeration system $(U_i)_{i\ge 0}$ satisfying, for $i\ge 1$, the relation
\begin{equation}
    \label{eq:fraenkel}
    U_{i}=(s+t-1)\, U_{i-1}+s\, U_{i-2}
\end{equation}
 where $U_{-1}=1/s$, $U_0=1$ and $s,t\in\mathbb{N}$ are positive.  The following proposition is easy to prove. It implies that the corresponding numeration system is a Pisot system.

\begin{proposition}
Let $s,t>0$ be integers.
    The positive root of $X^2-(s+t-1)X-s$ is a Pisot number. In particular, the numeration system defined by \eqref{eq:fraenkel} is a Pisot numeration system.
\end{proposition}

The take-away game devised by Fraenkel in \cite[Theorem~5.1]{Fra:98} has the following property: a pair $(x,y)$ is a $\mathcal{P}$-position if and only if
\begin{itemize}
  \item $\rep_U(x)$ ends in an even (possibly $0$) number of zeroes,
  \item $\rep_U(y)=\rep_U(x)0$.
\end{itemize}
This syntactical property permits us to state the following result.

\begin{proposition}\label{pro:st}
    For the take-away game in \cite{Fra:98}, the set of $\mathcal{P}$-positions is $U$-recognizable and thus, one can decide whether this set is $1$-invariant.
\end{proposition}

\begin{proof}
    Since we have a Pisot numeration system, the set $\mathbb{N}$ is $U$-recognizable \cite{BH}. Otherwise stated, there exists a DFA accepting $\rep_U(\mathbb{N})$. One can obviously modify this automaton to accept exactly the $U$-expansions ending with an even number of zeroes.

We can also construct a DFA recognizing pairs of words of the form $(0u,u0)$ where $u$ is a word over the alphabet $A$. The set of states is $A$. The initial state is $0$. This state is also final. States are used to store the last letter that was read on the second component.
Reading $(a,b)\in A^2$ from state $a$ leads to state $b$. Reading $(a,b)$ from state $c\neq a$ leads to a dead state. Such a DFA for $A=\{0,1\}$ is depicted in Figure~\ref{fig2}.

From these two automata, one can derive a DFA recognizing pairs of words of the form $(0u,u0)$ where $u$ is a $U$-expansion ending with an even number of zeroes. To finish the proof, one also has to consider the symmetric version for words of the form $(u0,0u)$.
\end{proof}

\begin{remark}
    For $s=t=1$, we are back to the Fibonacci sequence and Wythoff's game. The DFA accepting pairs of the forms $(0u,u0)$ where $u$ is a $F$-expansion ending with an even number of zeroes is depicted in Figure~\ref{fig:wyt}.
\begin{figure}[htbp]
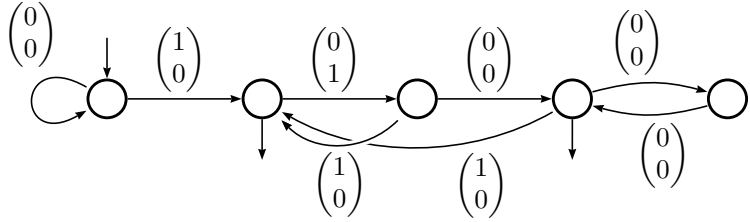

        \centering
\VCDraw{%
        \begin{VCPicture}{(-2,-2.5)(12,2)}
 \StateVar[]{(-1.7,0)}{1}
 \StateVar[]{(1.7,0)}{2}
 \StateVar[]{(5.1,0)}{3}
 \StateVar[]{(8.5,0)}{4}
 \StateVar[]{(11.9,0)}{5}
\Initial[n]{1}
\Final[s]{2}
\Final[s]{4}
\LoopW{1}{
  \begin{pmatrix}
      0\\ 0\\
  \end{pmatrix}
}

\EdgeL{1}{2}{
  \begin{pmatrix}
      1\\ 0\\
  \end{pmatrix}
}

\EdgeL{2}{3}{
  \begin{pmatrix}
      0\\ 1\\
  \end{pmatrix}
}

\EdgeL{3}{4}{
  \begin{pmatrix}
      0\\ 0\\
  \end{pmatrix}
}
\ArcL{4}{5}{
  \begin{pmatrix}
      0\\ 0\\
  \end{pmatrix}
}
\ArcL{5}{4}{
  \begin{pmatrix}
      0\\ 0\\
  \end{pmatrix}
}
\VArcL[.3]{arcangle=30,ncurv=.7}{4}{2}{
  \begin{pmatrix}
      1\\ 0\\
  \end{pmatrix}
}
\EdgeBorder
\VArcL[.5]{arcangle=45,ncurv=.8}{3}{2}{
  \begin{pmatrix}
      1\\ 0\\
  \end{pmatrix}
}
\EdgeBorderOff

\end{VCPicture}
}
\caption{The set of $\mathcal{P}$-positions of Wythoff's game is $F$-recognizable.}
        \label{fig:wyt}
    \end{figure}
\end{remark}

\subsection{Sets derived from $\mathcal{P}$-positions of existing games}

Here we will observe that many games like Wythoff's game, Tribonacci and Pisot unit game, Raleigh game, etc. have a $U$-recognizable set of $\mathcal{P}$-positions, and thus Theorem~\ref{the:main} can be applied.

$\bullet$ For Wythoff's game, with the strong syntactical properties of the $F$-representations of the $\mathcal{P}$-positions, we may apply Proposition~\ref{pro:st} and get the DFA depicted in Figure~\ref{fig:wyt}.

$\bullet$ A similar argument may be developed for the Tribonacci game because of the next statement. Note that this game has a set of $\mathcal{P}$-positions coded by $P_T$ where $T$ is the word given in \eqref{eq:tw}.

\begin{theorem}\cite[Theorem~5.2]{tribo}\label{the:tr}
The triple $(a,b,c)$ with $a<b<c$ is a $\mathcal{P}$-position of the Tribonacci game if and only if $\rep_T(a-1)=w0$, $\rep_T(b-1)=w01$ and $\rep_T(c-1)=w011$, where $\rep_T$ is the $T$-expansion associated with the Tribonacci sequence given in Example~\ref{exa:tribo}.
\end{theorem}

\begin{example} Table~\ref{tab:extri} illustrates the previous result. $A_m$ (respectively $B_m$ and $C_m$) denotes the position of the $m$th $1$ (resp. $2$ and $3$) in the Tribonacci word.
    \begin{table}[htbp]
$$\begin{array}{|ccc|rrr|}
\hline
A_m & B_m & C_m & \rep_T(A_m-1) &\rep_T(B_m-1) &\rep_T(C_m-1) \\
\hline
    1 & 2 &  4 &  \varepsilon & 1 & 11 \\
    3 & 6 & 11 &  10 & 101 & 1011 \\
5&9&17 & 100& 1001 & 10011 \\
7&13&24& 110& 1101 & 11011 \\
\hline
\end{array}$$
        \caption{First $\mathcal{P}$-positions and the corresponding Tribonacci-expansions.}
        \label{tab:extri}
    \end{table}
\end{example}

\begin{corollary}\label{cor:tri}
    It is decidable whether the set of $\mathcal{P}$-positions of the Tribonacci game \cite{tribo} is $1$-invariant.
\end{corollary}

\begin{proof}
    One can adapt the proof of Proposition~\ref{pro:st} showing that the set of $\mathcal{P}$-positions is $T$-recognizable. One has to intersect the DFA in Figure~\ref{figtri} with a DFA recognizing valid $T$-expansions for the first component. Note that $\rep_U(\mathbb{N})$ is the set of words over $\{0,1\}$ not containing the factor $111$. Also, for every Pisot numeration system $U$, the set $\{(i,i+1)\mid i\ge 0\}$ is $U$-recognizable \cite{Frougny,Frougny2}. Thus the difference of $1$ that appears in Theorem~\ref{the:tr} can easily be handled by finite automata. An alternative explanation is to make the following observation. Let $X$ be a $U$-recognizable subset of $\mathbb{N}^3$. By Theorem~\ref{the:equiv}, it is equivalent to the fact that $X$ is $U$-definable by a formula $\varphi$. Now, the set $X'$ defined by $(a,b,c)\in X\Leftrightarrow  (a+1,b+1,c+1)\in X'$ is trivially $U$-definable by a formula $\varphi'$ because the successor is definable in $\langle\mathbb{N},+\rangle$; see Remark~\ref{rem:var}. Indeed $\varphi'(a',b',c')$ holds if and only if there exist $a,b,c$ such that $\varphi(a,b,c)$ and $a'=a+1$, $b'=b+1$, $c'=c+1$.
\end{proof}

\begin{example}
In the DFA depicted in Figure~\ref{figtri}, reading $(a,b,c)$ from state $ab$ leads to state $bc$. States are used to store the last symbol read on the second and third component. We have duplicated the state $11$ to take into account that on the third component, we must accept words ending with $011$ and not those ending with $111$. Note that this DFA does not test the occurrence of a factor $111$.
\begin{figure}[htbp]
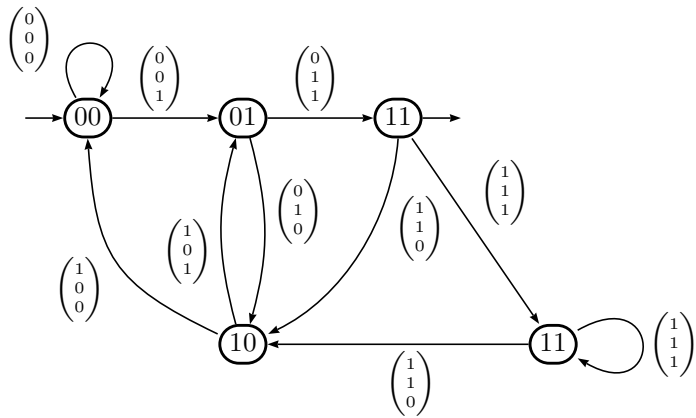

        \centering
\VCDraw{%
        \begin{VCPicture}{(-2,-6.5)(6,2)}
 \StateVar[00]{(-1.7,0)}{1}
 \StateVar[01]{(1.7,0)}{2}
 \StateVar[11]{(5.1,0)}{3}
 \StateVar[11]{(8.5,-5)}{4}
 \StateVar[10]{(1.7,-5)}{5}
\Initial[w]{1}
\Final[e]{3}
\LoopN{1}{
  {\tiny \begin{pmatrix}
      0\\0\\0\\
  \end{pmatrix}}
}
\LoopE[.5]{4}{
  {\tiny \begin{pmatrix}
      1\\1\\1\\
  \end{pmatrix}}
}
\EdgeL{1}{2}{
  {\tiny \begin{pmatrix}
      0\\0\\1\\
  \end{pmatrix}}
}
\EdgeL{2}{3}{
  {\tiny \begin{pmatrix}
      0\\1\\1\\
  \end{pmatrix}}
}
\EdgeL{3}{4}{
  {\tiny \begin{pmatrix}
      1\\1\\1\\
  \end{pmatrix}}
}
\VArcL[.5]{arcangle=30,ncurv=1}{5}{1}{
  {\tiny \begin{pmatrix}
      1\\0\\0\\
  \end{pmatrix}}
}
\ArcL{2}{5}{
  {\tiny \begin{pmatrix}
      0\\1\\0\\
  \end{pmatrix}}
}
\VArcL[.3]{arcangle=30,ncurv=.7}{3}{5}{
  {\tiny \begin{pmatrix}
      1\\1\\0\\
  \end{pmatrix}}
}
\EdgeL{4}{5}{
  {\tiny \begin{pmatrix}
      1\\1\\0\\
  \end{pmatrix}}
}
\ArcL{5}{2}{
  {\tiny \begin{pmatrix}
      1\\0\\1\\
  \end{pmatrix}}
}
\end{VCPicture}
}
\caption{A DFA recognizing $(00w0,0w01,w011)$.}
        \label{figtri}
    \end{figure}
\end{example}

$\bullet$ In \cite{pisot} the considered games have a set of $\mathcal{P}$-positions coded by generalized Tribonacci words that are a fixed point of the morphism $1\mapsto 1^s2, 2\mapsto 13, 3\mapsto 1$ where $s\ge 1$. We let $G_s$ denote this game (for a fixed value of the parameter $s$).
In this setting, one considers the sequence $(U_i)_{i\ge 0}$ where
\begin{equation}
    \label{eq:gen}
U_{i+3}=sU_{i+2}+U_{i+1}+U_i
\end{equation}
with the initial conditions $U_0=1$, $U_1=sU_0+1$, $U_2=sU_1+U_0+1$. The language of all $U$-expansions is recognized by the DFA depicted in Figure~\ref{fig:pisotg}

\begin{figure}[htbp]
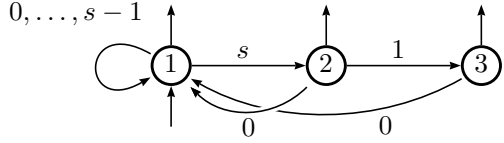

        \centering
\VCDraw{%
        \begin{VCPicture}{(-2,-2.5)(12,2)}

 \StateVar[1]{(1.7,0)}{2}
 \StateVar[2]{(5.1,0)}{3}
 \StateVar[3]{(8.5,0)}{4}

\Initial[s]{2}
\Final[n]{2}
\Final[n]{3}
\Final[n]{4}
\LoopW{2}{0,\ldots,s-1}
\EdgeL{2}{3}{s}
\EdgeL{3}{4}{1}
\VArcL[.3]{arcangle=30,ncurv=.7}{4}{2}{0}
\EdgeBorder
\VArcL[.5]{arcangle=45,ncurv=.8}{3}{2}{0}
\EdgeBorderOff

\end{VCPicture}
}
\caption{A DFA accepting all the $U$-representations.}
        \label{fig:pisotg}
    \end{figure}

The following theorem is a natural generalization of Theorem~\ref{the:tr}.
\begin{theorem}\cite[Section~4]{pisot}\label{the:pisot}
Let $s\ge 1$. The triple $(a,b,c)$ with $a<b<c$ is a $\mathcal{P}$-position of the game $G_s$ if and only if $\rep_U(a-1)=w$, $\rep_U(b-1)=ws$, $\rep_U(c-1)=ws1$ and $w$ is the label of a cycle starting from the initial state $1$ in the DFA depicted in Figure~\ref{fig:pisotg}, where $\rep_U$ is the $U$-expansion associated with the sequence \eqref{eq:gen}.
\end{theorem}

\begin{example}  Table~\ref{tab:gentri} illustrates the previous result. For $s=3$, the set of $\mathcal{P}$-positions of $G_3$ is coded by the word
$$111211121112131112111211121311121112111213111211112\cdots.$$
The numeration systems is constructed over the sequence $1,4,14,47,159,\ldots$.
    \begin{table}[htbp]
$$\begin{array}{|ccc|rrr|}
\hline
A_m & B_m & C_m & \rep_T(A_m-1) &\rep_T(B_m-1) &\rep_T(C_m-1) \\
\hline
1&4&14&\varepsilon&3&31\\
2&8&28&1&13&131\\
3&12&42&2&23&231\\
5&18&61&10&103&1031\\
\hline
\end{array}$$
        \caption{First $\mathcal{P}$-positions of $G_3$ and the corresponding expansions.}
        \label{tab:gentri}
    \end{table}
\end{example}

\begin{corollary}
  Let $s\ge 1$.  It is decidable whether the set of $\mathcal{P}$-positions of the game $G_s$ \cite{pisot} is $1$-invariant.
\end{corollary}

\begin{proof}
    Again the proof is similar to the one of Proposition~\ref{pro:st} and Corollary~\ref{cor:tri}. One has to devise a DFA for words of the form $(00u,0u3,u31)$. Moreover, testing if $u$ labels a cycle can also be handled by a DFA. Finally, the positive root of $X^3-sX^2-X-1$ is again  a Pisot number, so  the set $\{(i,i+1)\mid i\ge 0\}$ is $U$-recognizable \cite{Frougny,Frougny2}.
\end{proof}

$\bullet$ The Raleigh game \cite{Raleigh} is a variant game played on three piles of tokens. Again this game has a set of $\mathcal{P}$-positions that is $F$-recognizable. First we provide a new (morphic) characterization of its set of $\mathcal{P}$-positions.

\begin{lemma}
    The set of $P$-positions of the Raleigh game is coded by the fixed point $12312123\cdots$ of the morphism $1\mapsto 12, 2\mapsto 3,3\mapsto 12$.
\end{lemma}

\begin{proof}
    Let $w$ be the fixed point of the morphism $f:1\mapsto 12, 2\mapsto 3,3\mapsto 12$.
    In \cite{Raleigh}, the $\mathcal{P}$-positions of the Raleigh game ($A_m$, $B_m$ and $C_m$ are defined as before) are characterized according to the following relations: $A_1=1$, $B_1=2$, $C_1=3$ and for all $m>1$, 
    \begin{eqnarray*}
    A_m&=&mex\{A_i,B_i,C_i:0\leq i<m\}\\
    B_m&=&A_m+1\\
    C_m&=&\left\{\begin{array}{cl}
    C_{m-1}+3,&\mbox{if $A_{m}-A_{m-1}=2$};\\
    C_{m-1}+5,&\mbox{otherwise}.\end{array}\right.
    \end{eqnarray*}
    It is not hard so see that $w$ is the unique word over $\{1,2,3\}$ which exactly corresponds to the sequence $(A_m,B_m,C_m)$. Indeed, the definition of $f$ implies that for all $m\geq 1$, the $m$th $1$ of $w$ appears before the $m$th $2$ and the $m$th $3$. In other words, the $1$'s in $w$ correspond to the values $A_m$. By definition of $f$, one can also observe that the $m$th $1$ and the $m$th $2$ are always successors in $w$. Thus the $2$'s of $w$ correspond to $B_m$. Concerning the $3$'s of $w$, there are only ``produced'' as images of values $2$. In other words, the $m$th $2$ produces the $m$th $3$. Since the gap between two $2$'s is either $2$ or $3$, and as $|f(1)|=|f(3)|=2$, the difference between the $m$th and $(m+1)$st $3$ is equal to $3$ or $5$, according to the gap between the $m$th and $(m+1)$st $2$ (which is identical to the gap between the $m$th and $(m+1)$st $1$).
\end{proof}

 The DFA associated with the morphism $f:1\mapsto 12, 2\mapsto 3,3\mapsto 12$ is depicted in Figure~\ref{fig:raleigh}. It is defined as follows (and this definition can be extended to any morphism). Its set of states is $\{1,2,3\}$. If $f(a)=bc$, then the DFA has an edge from $a$ to $b$ (resp., $c$) of label $0$ (resp., $1$). If $f(a)=b$, then the DFA has an edge from $a$ to $b$ of label $0$. The initial state is $1$ because the word is obtained by iterating $f$ from $1$. Note that this DFA has a synchronizing property. Reading $00$ (resp., $01$, $010$) from every state leads to state $1$ (resp., $2$, $3$).
\begin{figure}[h!tbp]
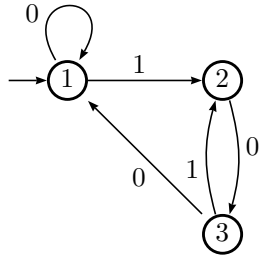

        \centering
\VCDraw{%
        \begin{VCPicture}{(-1.7,-0.5)(3.4,5)}
 \StateVar[1]{(-1.7,3.4)}{b}
 \StateVar[2]{(1.7,3.4)}{a}
 \StateVar[3]{(1.7,0)}{c}
\Initial[w]{b}
\LoopN{b}{0}
\EdgeL{b}{a}{1}
\ArcL{a}{c}{0}
\ArcL{c}{a}{1}
\EdgeL{c}{b}{0}
\end{VCPicture}
}
\caption{The DFA associated with the Raleigh morphism.}
        \label{fig:raleigh}
    \end{figure}

 We recall the following general result. See, for instance, \cite[Section~3.4]{cant}
\begin{proposition}\cite{Maes}\label{pro:maes}
    Lat $A$ be a finite alphabet. Let $f:A^*\to A^*$ be a morphism generating an infinite word $w$ when iterating $f$ on a symbol $a\in A$.
Let $M$ be the language of words not starting with $0$ and accepted by automaton associated with $f$ where $a$ is the initial state and all states are final. The $m$th symbol of $w$ (where indexing start with $0$) is the state reached from $a$ when reading the $(m+1)$st word of $M$ genealogically ordered.
\end{proposition}

\begin{example}
    Observe that the DFA depicted in Figure~\ref{fig:raleigh} recognizes exactly the $F$-expansions of the integers. The first few words (not starting by zero) accepted by this automaton (the first elements in $M$) are given below. When reading these words from state $1$, we have indicated the corresponding reached state.
    $$\begin{array}{cccccccc}
        \varepsilon & 1 & 10 & 100 & 101 & 1000 & 1001 & 1010 \\
1&2&3&1&2&1&2&3\\
    \end{array}$$
The reader may observe that we have a prefix of the infinite word generated by $f$. So to get the $m$th symbol (where the first symbol has index $0$) of this word, one has to feed the DFA with $\rep_F(m)$.
\end{example}

\begin{theorem}\label{the:ral}
     The $3$-tuple $(a,b,c)$ with $a<b<c$ is a $\mathcal{P}$-position of the Raleigh game if and only if $\rep_F(a-1)=w00$, $\rep_F(b-1)=w01$, $\rep_F(c-1)=w010$ where $\rep_F$ is the $F$-expansion associated with the Fibonacci sequence.
\end{theorem}

\begin{example}Table~\ref{tab:ral} illustrates the previous result. It is pretty easy to enumerate the $\mathcal{P}$-positions. Enumerate (by genealogical ordering) the words over $\{0,1\}$ avoiding the factor $11$ and ending with $0$: $10,100,1000,1010,\ldots$. For the $m$th element $w$ in this enumeration, consider the $3$-tuple of words $(w0,w1,w10)$.  Note that the three added suffixes all start with $0$.
\begin{table}[h]
    $$\begin{array}{|ccc|rrr|}
\hline
A_m & B_m & C_m &  \rep_F(A_m-1) &\rep_F(B_m-1) &\rep_F(C_m-1) \\
\hline
1&2&3&\varepsilon&1&10\\
4&5&8&100&101&1010\\
6&7&11&1000&1001&10010\\
9&10&16&10000&10001&100010\\
12&13&21&10100&10101&101010\\
\hline
\end{array}$$
        \caption{First $\mathcal{P}$-positions of the Raleigh game and the corresponding $F$-expansions.}
        \label{tab:ral}
    \end{table}
\end{example}

\begin{proof}
 We make use of Proposition~\ref{pro:maes}.  Let $M$ be the language made of the words not starting with $0$ and accepted by the automaton depicted in Figure~\ref{fig:raleigh} where all states are final. Note that $M=\rep_F(\mathbb{N})$. Hence the state reached when reading $\rep_F(m)$ provides the $m$th symbol of the word (index starting with $0$).

 The set $M\cap \{0,1\}^{\ge 4}$ is partitioned into four sets $F_1,F_2,F_3$ of words of length at least four:  those ending $00$, $01$ and $010$ respectively. Thus for all $m\ge 5$, $\rep_F(m)$ belongs to one of these sets. Moreover, if $\rep_F(m)$ belongs to $F_i$, then the $m$th symbol of the word is $i$.

 Since $M$ is prefix-closed, we know that if $w$ and $w'$ are two consecutive words in $M$ of length at least two (for the genealogical ordering), then $w00$ and $w'00$ are consecutive elements in $F_1$. From Proposition~\ref{pro:maes}, they correspond to two consecutive occurrences of the symbol $1$ in the infinite word. The same observation is made for elements in $F_2,F_3$ (for the words $w01$, $w'01$ and $w010$, $w'010$ respectively).

 From this observation, to finish the proof by induction --- enumerating by increasing genealogical order the words of $M$ --- we just need to find one $3$-tuple $(a,b,c)$ which is a $\mathcal{P}$-position of the Raleigh game
 and such that $\rep_F(a-1,b-1,c-1)=(u00,u01,u010)$. One can take $(6,7,11)$.  We have $$\rep_F(5,6,10)=(u00,u01,u010)\text{ with }u=10.$$
 The first few values have thus to be checked by direct inspection.
\end{proof}

\begin{corollary}\label{cor:tri}
    It is decidable whether the set of $\mathcal{P}$-positions of the Raleigh game is $1$-invariant.
\end{corollary}

\begin{remark}
    One can redo the proof of Theorem~\ref{the:pisot} in a way similar to the proof of Theorem~\ref{the:ral}. Let $M$ be the language made of the words not starting with $0$ and accepted by the automaton depicted in Figure~\ref{fig:pisotg}. The set $M\cap \{0,1\}^{\ge 2}$ is partitioned into three sets $F_1,F_2,F_3$ of words of length at least two: those not ending with $s$ nor $s1$, those ending with $s$, and those ending with $s1$ respectively.
\end{remark}


\begin{thebibliography}{99}

  \bibitem{Bo} N.~B. Ho, Two variants of Wythoff's game preserving its $\mathcal{P}$-positions, {\em J. Combin. Theory Ser. A} {\bf 119} (2012), no. 6, 1302--1314.

  \bibitem{cant} V. Berth{\'e}, M. Rigo (Eds.), {\em Combinatorics, automata and number theory}, Encyclopedia of Mathematics and its Applications {\bf 135}, Cambridge University Press, Cambridge, 2010.

  \bibitem{BH} V. Bruy{\`e}re, G. Hansel, Bertrand numeration systems and recognizability, {\em Theoret. Comput. Sci.} {\bf 181} (1997), no. 1, 17---43.

  \bibitem{BHMV} V. Bruy{\`e}re, G. Hansel, C.~Michaux, R.~Villemaire, Logic and $p$-recognizable sets of integers, {\em Bull. Belg. Math. Soc. Simon Stevin} {\bf 1} (1994), no. 2, 191--238.

  \bibitem{Buc} J. R.~B{\"u}chi, Weak second-order arithmetic and finite automata, {\em Z. Math. Logik Grundlagen Math.} {\bf 6} (1960), 66-–92.

  \bibitem{Carpi} A.~Carpi, C.~Maggi, On synchronized sequences and their separators,
{\em Theor. Inform. Appl.} {\bf 35} (2001), no. 6, 513--524 (2002).

\bibitem{CaDR} J. Cassaigne, E. Duch\^ene, M. Rigo, Invariant games and non-homogeneous Beatty sequences, preprint.

  \bibitem{aut1} C. F. Du, H. Mousavi, L. Schaeffer, J. Shallit, Decision Algorithms for Fibonacci-Automatic Words, with Applications to Pattern Avoidance , {\tt arXiv:1406.0670}

  \bibitem{FraNow} E.~Duch{\^e}ne, A.S.~Fraenkel, R.J.~Nowakowski, M.~Rigo, Extensions and restrictions of Wythoff's game preserving its $\mathcal{P}$ positions, {\em J. Combin. Theory Ser. A} {\bf 117} (2010), no. 5, 545-–567.

  \bibitem{tribo} E.~Duch{\^e}ne, M.~Rigo, A morphic approach to combinatorial games: the Tribonacci case, {\em Theor. Inform. Appl.} {\bf 42} (2008), no. 2, 375–-393.

  \bibitem{pisot} E.~Duch{\^e}ne, M.~Rigo, Pisot unit combinatorial games. Monatsh. Math. 155 (2008), no. 3-4, 217--249.

  \bibitem{DucRig} E.~Duch{\^e}ne, M.~Rigo, Invariant games, {\em Theoret. Comput. Sci.} {\bf 411} (2010), no. 34-36, 3169--3180.

  \bibitem{Fra:85} A.S.~Fraenkel, Systems of numeration, {\em Amer. Math. Monthly} {\bf 92} (1985), no. 2, 105--114.

  \bibitem{Fra:98} A.S.~Fraenkel, Heap games, numeration systems and sequences, {\em Ann. Comb.} {\bf 2} (1998), no. 3, 197--210.


     \bibitem{mark} A.S. Fraenkel, Aperiodic subtraction games, {\em Electronic J. Combinatorics} {\bf 18} (2011), no. 2, 19--31.

  \bibitem{Rat} A.S.~Fraenkel, The Rat Game and the Mouse Game, preprint.

  \bibitem{Raleigh} A.S.~Fraenkel, The Raleigh Game, INTEGERS {\bf 7} (2007) \#13, 11pp.

  \bibitem{FraLar} A.S.~Fraenkel, U.~Larsson, Take-away games on Beatty's theorem and the notion of invariance, preprint.

  \bibitem{Frougny} Ch. Frougny, Representations of numbers and finite automata, {\em Math. Systems Theory} {\bf 25} (1992), no. 1, 37--60.

  \bibitem{Frougny2} Ch. Frougny, On the sequentiality of the successor function, {\em Inform. and Comput.} {\bf 139} (1997), no. 1, 17--38.

  \bibitem{goc} D. Go\v{c}, D. Henshall, J. Shallit, Automatic theorem-proving in combinatorics on words, In N.~Moreira and R.~Reis, editors, CIAA 2012, {\em Lect. Notes in Computer Science} {\bf 7381}, 180--191. Springer-Verlag, 2012.

  \bibitem{goc2} D. Go\v{c}, N. Rampersad, M. Rigo, P. Salimov, On the Number of Abelian Bordered Words (with an example of automatic theorem-proving), to appear in {\em Internat. J. Found. Comput. Sci.}.

       \bibitem{Lar} U. Larsson, P. Hegarty, A. S. Fraenkel, Invariant and
    dual subtraction games resolving the Duch{\^e}ne--Rigo conjecture,
    {\it Theoret. Comput. Sci.} {\bf 412} (2011), 729--735.

  \bibitem{Lot} M. Lothaire, {\em Combinatorics on words}, Corrected reprint of the 1983 original, Cambridge Mathematical Library, Cambridge University Press, Cambridge, 1997.

  \bibitem{aut2} H. Mousavi, J. Shallit, Mechanical Proofs of Properties of the Tribonacci Word, {\tt arXiv:1407.5841}

  \bibitem{Maes} M. Rigo, A. Maes, More on generalized automatic sequences, {\em J. Autom. Lang. Comb.} {\bf 7} (2002), 351-–376.

  \bibitem{flans} M. Rigo, {\em Formal Languages, Automata and Numeration Systems}, to appear.

  \bibitem{saka} J. Sakarovitch, {\em Elements of automata theory}, Cambridge University Press, Cambridge, (2009).
\end{thebibliography}
\end{document}